\documentclass{amsart}
\usepackage{amssymb,amsmath,amsfonts}
\newtheorem{theorem}{Theorem}[section]
\newtheorem{lemma}[theorem]{Lemma}
\newtheorem{proposition}[theorem]{Proposition}
\newtheorem{corollary}[theorem]{Corollary}

\newtheorem{rmq}[theorem]{Remark}

\newtheorem{definition}[theorem]{Definition}

\newcommand \A{\mathcal{A}}
\newcommand \F{\mathcal{F}}

\numberwithin{equation}{section}

\title[]{Call option on the maximum of the interest rate in the one factor affine model} 

\author{Mohamad Houda}

\address{Normandie Univ., Laboratoire Rapha\"el Salem,
UMR CNRS 6085, Rouen, France}
\email{mhamed.machhour@gmail.com}

\date{September 09, 2013}

\begin{document}

\maketitle
\setcounter{section}{0}
\begin{abstract}
We determine an explicit formula for the Laplace transform of the price of an option on a maximal interest rate when the instantaneous rate satisfies Cox-Ingersoll-Ross's model. This generalizes considerably one result of Leblanc-Scaillet \cite{LeblancScaillet1998}.

\vspace{2mm}
Keywords: affine model, hitting time, Laplace transform.
\end{abstract}

\section{Introduction}
The models of the instantaneous rate are mainly used to price and cover the discount bonds and the options on the discount bonds. Up to now, no model has been able to triumph as a reference model such as Black-Scholes's model for the options on assets. In this paper, our aim is to price interest rate option namely European path dependent option on yields. More precisely, we are interested in a call on a maximum. For this type of option, we give an analytical pricing formula.

In order to answer our pricing problem, we consider the affine class of one factor term structure models with time invariant parameter studied in a more general framework by Duffie and Kan \cite{DuffieKan1996}. Particular cases are the Cox et al. (CIR) (1985) and Vasicek (1977) models. In these models, the yield of the discount bond is an affine function of the instantaneous rate.

In section 2 an analytical formula for a European path dependent option on yields is derived. We examine an option on maximum and we discuss the problem of the numerical implementation of this formula. Some pratical results for the European call option on maximum are presented. A proof of Novikov's condition and a simulation program are gathered in appendices.

Let us first give some definitions and recall some preliminary results.
\begin{definition}{\em\cite[$p.504$ $(13.1.1)$]{AbramowitzStegun1965}}\label{Kummer}
For $c \in \mathbf{R}$ and $b\in \mathbf R$, such that $-b\not\in \mathbf N$, \textit{Kummer's function} is defined by 
$$
M(c, b, z) = 1 +\frac{cz}{b} +\frac{(c)_2 z^2}{(b)_2 2!} + ........+\frac{(c)_n z^n}{(b)_n n!} + ....
$$
where,
$
(c)_n = c(c+1)(c+2).....(c+n-1) ,\,(c)_0 = 1 
$.
\end{definition}
\begin{rmq}
The function $M$ is analytical on $\mathbf R$ and satisfies \textit{Kummer's equation}
\begin{equation}\label{eq.K}
z\frac{d^2w}{dz^2} + (b-z)\frac{dw}{dz} - c\,w = 0.
\end{equation}
As $M$ is analytical, it is bounded in a neighborhood of $0$.
\end{rmq}
\vspace{2mm}
Let $(\Omega, \F, (\F_t)_{t\geq 0}, Q)$ be a filtered probability space, and let $(w(t))_{t\geq 0}$ be an $(\F_t)_{t\geq 0} $ standard Brownian motion. Let $y_t$ be the solution of the stochastic differential equation 
\begin{equation}\label{eds}
d y_t = \mu(y_t)\,dt \, + \sigma(y_t)\,d w(t),
\end{equation}
starting from $y_0$.
\begin{theorem}{\em\cite[$p.354$]{LeblancScaillet1998}}\label{L.S} 
Let $T^{(y)}_a = \inf\{t; y_t \geq a\}$, and let $V$ a bounded function on $[0, a]$, such that $\A(V) = \gamma V$, where 
$$
\A = \frac{1}{2}\sigma^2(y)\frac{d^2}{dy^2} + \mu(y)\frac{d}{dy}
$$
is the infinitesimal generator associated of \ref{eds}. Then, when $y_0 \leq a$, we have 
$$
E^Q_{y_0}[e^{-\gamma T^{(y)}_a} 1_{T^{(y)}_a < \infty}] = \frac{V(y_0)}{V(a)}.
$$
\end{theorem}
\section{one factor affine model}
We denote by $B(t, t+\tau)$ the price at time $t$ of a discount bond of maturity $t +\tau$, i.e. the price of the asset delivering one euro at time $t+\tau$ ($\tau$ is independent of $t$).

The yield corresponding to this bond $Y(t, t+\tau)$ at time $t$ with maturity $\tau$ is defined by
$$
Y(t, t+\tau) = -\frac{1}{\tau}\log B(t, t+\tau).
$$ 
Assume that, for each $\tau$, the yield $Y(t, t+\tau)$ is an affine function of the instantaneous interest rate $r(t)$ :
$$
Y(t, t+\tau) = \frac{1}{\tau}[A(\tau)r(t) + b(\tau)]\,,
$$  
and that $r(t)$ is also a solution of the stochastic differential equation of type \eqref{eds}. In this case, under the risk neutral probability $Q$, the instantaneous interest rate $r(t)$ satisfies 
the stochastic differential equation
\begin{equation}\label{eaffine}
d r(t) = (\phi - \lambda r (t))\,dt + \sqrt{\alpha r(t) + \beta} \, d w(t) ,
\end{equation}
starting from $r(0)=r_0$ \cite{DuffieKan1996}.
\section{Option european on the maximum}
The price at date $0$ of a European call option on maximum of maturity $T$ and strike price $K$ is given by 
$$
C(\sup_{u\in[0,T]}Y(u, u+\tau),0,T,K) = E^Q_{r_0}[e^{-\int_0^Tr_s ds}(\sup_{u\in[0,T]}Y(u,u+\tau)-K)_+]\,.
$$ 
By using the affine form $Y(t, t+\tau)$ with respect to $r_t$, the above formula can be written as 
$$
C(\sup_{u\in[0,T]} Y(u, u+\tau),0,T,K) = \frac{A(\tau)}{\tau}\,C(\sup_{u\in[0,T]} r_u,0,T,k),
$$
where $C(\sup_{u\in[0,T]} r_u,0,T,k)$ represents the price of European call on the instantaneous rate with a strike $k = \frac{\tau K -b(\tau)}{A(\tau)}$.
\begin{theorem}\label{Th3.1}
Let $r_t$ the solution of \eqref{eaffine}, then
$$
E^Q_{r_0}[e^{-\gamma T^{(r)}_a}] = \frac{M(\frac{\gamma}{\lambda},(\phi + \frac{\lambda\beta}{\alpha})\frac{2}{\alpha} , \frac{2\lambda(r_0+\beta/\alpha)}{\alpha})}{M(\frac{\gamma}{\lambda},(\phi + \frac{\lambda\beta}{\alpha})\frac{2}{\alpha} , \frac{2\lambda (a + \beta/\alpha)}{\alpha})}.
$$
\end{theorem}
\begin{proof}
By using the change of variable $\tilde{r}(t) = r(t) + \frac{\beta}{\alpha}$ ; then $ \tilde{r}(t)$ satisfies
\begin{eqnarray}\label{eq.3.1}
d \tilde{r}(t) = (\phi + \frac{\lambda \beta}{\alpha}- \lambda \tilde{r}(t))\,dt + \sqrt{\alpha \tilde{r}(t)} \, d w(t)\,.
\end{eqnarray}
In this case, the infinitesimal generator of $\tilde{r}$ process is
$$
\tilde{\A}= (\phi + \frac{\lambda \beta}{\alpha}- \lambda \tilde{r})\frac{d}{d\tilde{r}} + \frac{\alpha}{2} \tilde{r} \frac{d^2}{d\tilde{r}^2}\,.
$$
Let $V(x) = M(\frac{\gamma}{\lambda} , \frac{2}{\alpha}(\phi + \frac{\lambda \beta}{\alpha}) , \frac{2\lambda x}{\alpha})$. The function $V$ is bounded on $[0, a + \frac{\beta}{\alpha}]$ (by the properties of Kummer's equation \eqref{eq.K}), and
$$
\tilde{\A} V(x) = (\phi + \frac{\lambda \beta}{\alpha}- \lambda x)V'(x) + \frac{\alpha}{2} x V''(x) . 
$$
Let us denote $\tilde{\phi} = \phi + \frac{\lambda \beta}{\alpha}$. By using the change of variable $z = \frac{2\lambda x}{\alpha} $, we get
\begin{eqnarray*}
\tilde{\A} V(x) &=& (\tilde{\phi}- \lambda x)\, \frac{\partial M}{\partial z} (\frac{\gamma}{\lambda}, \frac{2 \tilde{\phi}}{\alpha}, \frac{2\lambda x}{\alpha})\,\frac{2\lambda}{\alpha} + \frac{\alpha x}{2} \,\frac{\partial^2 M}{\partial z^2}(\frac{\gamma}{\lambda} , \frac{2 \tilde{\phi}}{\alpha} , \frac{2\lambda x}{\alpha})\,\frac{4\lambda^2}{\alpha^2} \\
&=& \lambda z \frac{\partial^2 M}{\partial z^2}(\frac{\gamma}{\lambda} , \frac{2 \tilde{\phi}}{\alpha} , \frac{2\lambda x}{\alpha}) + (\frac{2\lambda \tilde{\phi}}{\alpha} - \lambda z) \frac{\partial M}{\partial z}(\frac{\gamma}{\lambda}, \frac{2 \tilde{\phi}}{\alpha}, \frac{2\lambda x}{\alpha})\\
&=& \lambda\, \frac{\gamma}{\lambda} M(\frac{\gamma}{\lambda}, \frac{2 \tilde{\phi}}{\alpha}, \frac{2\lambda x}{\alpha})\,(car\, M\, satisfait (5.1))\\
&=& \gamma\, V(x).
\end{eqnarray*}
So we can apply Theorem \ref{L.S} on $V$, $\tilde{r}$, $\mu(\tilde{r}(t)) = (\phi + \frac{\lambda\beta}{\alpha} - \lambda \tilde{r}(t))$, $\sigma(\tilde{r}(t)) = \sqrt{\alpha \tilde{r}(t)}$ and $\tilde{r}(0) = r_0 + \frac{\beta}{\alpha} $, and therefore, since $T^{(r)}_a = T^{(\tilde{r})}_{a+\frac{\beta}{\alpha}} $, we have
$$
E^Q_{r_0}[e^{-\gamma T^{(r)}_a} 1_{T^{(r)}_a < \infty}]=\frac{V(r_0 + \frac{\beta}{\alpha})}{V(a + \frac{\beta}{\alpha})} = \frac{M(\frac{\gamma}{\lambda} , \frac{2}{\alpha}(\phi + \frac{\lambda \beta}{\alpha}) , \frac{2\lambda (r_0 + \beta/\alpha)}{\alpha})}{M(\frac{\gamma}{\lambda} , \frac{2}{\alpha}(\phi + \frac{\lambda \beta}{\alpha}) , \frac{2\lambda (a + \beta/\alpha)}{\alpha})}.
$$
Letting $\gamma$ go to zero , since $M(0, b, z) = 1$ , we find 
\begin{eqnarray*}
Q (T^{(r)}_a < \infty) &=& E^Q_{r_0}[1_{T^{(r)}_a < \infty}]\\ 
&=& \lim_{\gamma \rightarrow 0^+}
E^Q_{r_0}[e^{-\gamma T^{(r)}_a} 1_{T^{(r)}_a < \infty}]\\ 
&=& \lim_{\gamma \rightarrow 0^+} \frac{M(\frac{\gamma}{\lambda} , \frac{2}{\alpha}(\phi + \frac{\lambda \beta}{\alpha}) , \frac{2\lambda (r_0 + \beta/\alpha)}{\alpha})}{M(\frac{\gamma}{\lambda} , \frac{2}{\alpha}(\phi + \frac{\lambda \beta}{\alpha}) , \frac{2\lambda (a + \beta/\alpha)}{\alpha})}\\ 
&=& \frac{M(0 , \frac{2}{\alpha}(\phi + \frac{\lambda \beta}{\alpha}) , \frac{2\lambda (r_0 + \beta/\alpha)}{\alpha})}{M(0 , \frac{2}{\alpha}(\phi + \frac{\lambda \beta}{\alpha}) , \frac{2\lambda (a + \beta/\alpha)}{\alpha})}\\
&=& 1\,,
\end{eqnarray*}
which gives us $T^{(r)}_a < \infty$, $Q$-almost surely. Then
$$ 
E^Q_{r_0}[e^{-\gamma T^{(r)}_a}] = E^Q_{r_0}[e^{-\gamma T^{(r)}_a} 1_{T^{(r)}_a < \infty}] = \frac{M(\frac{\gamma}{\lambda} , \frac{2}{\alpha}(\phi + \frac{\lambda \beta}{\alpha}) , \frac{2\lambda (r_0 + \beta/\alpha)}{\alpha})}{M(\frac{\gamma}{\lambda} , \frac{2}{\alpha}(\phi + \frac{\lambda \beta}{\alpha}) , \frac{2\lambda (a + \beta/\alpha)}{\alpha})}.
$$
\end{proof}
\subsection{Option price}

We give in the following proposition an explicit formula to the Laplace transform of a European call option on the maximum of the instantaneous rate .
\begin{proposition}\label{prop3.2}
For all $\tilde{a} > \max(0, \frac{\beta}{\alpha} - \frac{\theta}{\sqrt{\alpha}}(\phi + \frac{\lambda \beta}{\alpha}))$ and all $r_0\leq k$, we have 
$$
U_{(k,r_0)}(\tilde{a})= e^{-\frac{\theta r_0 }{\sqrt{\alpha}}}\,M(\frac{\tilde{\gamma}}{\tilde{\lambda}},(\tilde{\phi}_{0} + \frac{\tilde{\lambda}\beta}{\alpha})\frac{2}{\alpha} , \frac{2\tilde{\lambda}(r_0+\beta/\alpha)}{\alpha})
\int_k^\infty \frac{e^{\frac{\theta v}{\sqrt{\alpha}}} P_{\tilde{a}}(v)}{M(\frac{\tilde{\gamma}}{\tilde{\lambda}},(\tilde{\phi}_{0} + \frac{\tilde{\lambda}\beta}{\alpha})\frac{2}{\alpha} , \frac{2\tilde{\lambda}(v + \beta/\alpha)}{\alpha})}dv\,,
$$
where $M$ is the Kummer's function defined in \ref{Kummer}, and for all $\tilde{a}$, $P_{\tilde{a}}(v):=  \int_0^\infty dTe^{-\tilde{a}T} B_v(0,T)$.
\end{proposition}
\begin{proof}
Note that:
$$
(\sup_{u\in[0,T]}r_u -k)_+ = \int_k^\infty 1_{\sup_{[0,T]} r_u > v}\, dv\,,
$$
and
$
T_v = inf\{t; r_t > v \} = T^{(r)}_v.
$

Let $U_{(k,r_0)}$ be the Laplace transform of $C(\sup_{u\in[0,T]} r_u , 0 , T , k)$ considered as function of $T$, then for all $\tilde{a} >\max(0, \frac{\beta}{\alpha} - \frac{\theta}{\sqrt{\alpha}}(\phi + \frac{\lambda \beta}{\alpha}))$,
\begin{eqnarray*}
U_{(k,r_0)}(\tilde{a})
&=&\int_0^\infty e^{-\tilde{a} T} C(\sup_{u\in[0,T]} r_u , 0 , T , k)dT \\
&=& E^Q_{r_0}[\int_k^\infty dv\,\int_0^\infty dT e^{-\tilde{a}T-\int_0^T r_s ds} 1_{\sup_{[0,T]} r_u > v}]\\
&=& E^Q_{r_0}[\int_k^\infty dv\,\int_{T_v}^\infty dT e^{-\tilde{a}T-\int_0^T r_s ds}]\\
&=& E^Q_{r_0}[\int_k^\infty dv\,\int_{T_v}^\infty dT e^{-\tilde{a}(T-T_v)-\tilde{a} T_v - \int_0^{T_v} r_s ds - \int_{T_v}^T r_s ds}]\\
&=& \int_k^\infty dv E^Q_{r_0}[e^{-\tilde{a} T_v - \int_0^{T_v} r_s ds} E^Q_{r_0}[\int_{T_v}^\infty dT e^{-\tilde{a}(T-T_v)-\int_{T_v}^T r_s ds} / T_v]]\\
&=& \int_k^\infty dv E^Q_{r_0}[e^{- \tilde{a}T_v - \int_0^{T_v} r_s ds}] E^Q_v[\int_0^\infty dTe^{-\tilde{a} T-\int_0^T r_s ds}].
\end{eqnarray*}
The last equality is due to the strong Markov property and the time homogeneity of the instantaneous interest rate process.

First, we will determine : $E^Q_{r_0}[e^{-\tilde{a} T_v - \int_0^{T_v} r_s ds}]$. By the change of variable: $\tilde{r}_t = r_t + \frac{\beta}{\alpha}$, we have 
$$
E^Q_{r_0}[e^{-\tilde{a} T_v - \int_0^{T_v} r_s ds}] 
= 
E^Q_{(r_0 + \frac{\beta}{\alpha})}[e^{-(\tilde{a}-\frac{\beta}{\alpha}) T_v - \int_0^{T_v} \tilde{r}_s ds}]\,.
$$
Let us now do the change of measure
\begin{eqnarray*}
\frac{dQ^*}{dQ} &=& \exp\left(-\int_0^t \theta \sqrt{\tilde{r}(s)}\, d w(s) - \frac{1}{2}\int_0^t\theta^2 \tilde{r}_s ds
\right)\\
&=&\exp\left(-\frac{1}{\sqrt{\alpha}}\int_0^t \theta 
[d\tilde{r}_s - (\phi + \frac{\lambda \beta}{\alpha}- \lambda \tilde{r}_s)\,ds] 
-\frac{1}{2}\int_0^t\theta^2 \tilde{r}_s ds
\right)\,;
\end{eqnarray*}
this is due of equation \eqref{eq.3.1} for $\tilde{r}(t)$. Here, Novikov's condition : $E^Q_{\tilde{r}_0}(e^{\frac{1}{2}\int_0^T\theta^2 \tilde{r}_s ds}) < \infty $ is satisfied and it is the subject of Lemma \ref{Le} (see Appendix on Novikov's condition).

By Girsanov'Theorem, $Q^*$ is a probability, and therefore 
\begin{eqnarray*}
E^Q_{r_0}[e^{-\tilde{a} T_v - \int_0^{T_v} r_s ds}] &=& E^Q_{(r_0 + \frac{\beta}{\alpha})}[e^{-(\tilde{a}-\frac{\beta}{\alpha}) T_v - \int_0^{T_v} \tilde{r}_s ds}]
\\
&=& E^{Q*}_{(r_0 + \frac{\beta}{\alpha})}
\left[
e^{-(\tilde{a}-\frac{\beta}{\alpha}) T_v - \int_0^{T_v} \tilde{r}_s ds}\,e^{ \left(\frac{1}{\sqrt{\alpha}}\int_0^{T_v} \theta 
[d\tilde{r}_s - (\phi + \frac{\lambda \beta}{\alpha}- \lambda \tilde{r}_s)\,ds] 
+ \frac{1}{2}\int_0^{T_v}\theta^2 \tilde{r}_s ds
\right)
}
\right]
\\
&=& E^{Q*}_{(r_0 + \frac{\beta}{\alpha})}
\left[
e^{-(\tilde{a}-\frac{\beta}{\alpha}) T_v - \int_0^{T_v} \tilde{r}_s ds} 
\,e^{\frac{\theta}{\sqrt{\alpha}}(v - {r}_0) - \frac{\theta}{\sqrt{\alpha}}((\phi + \frac{\lambda \beta}{\alpha})T_v} \,e^{\int_0^{T_v}(\frac{\lambda\theta}{\sqrt{\alpha}} + \frac{\theta^2}{2})\tilde{r}_s ds}
\right].
\end{eqnarray*}
We choose $\theta$ as : $\frac{\lambda\theta}{\sqrt{\alpha}} + \frac{\theta^2}{2} = 1$ ($\alpha > 0$ and $\lambda \in \mathbf R^* $) , this gives us :

$\sqrt{\alpha}\theta^2 + 2\lambda\theta - 2\sqrt{\alpha} = 0 $, and
$$
\theta = \frac{-\lambda \pm \sqrt{\lambda^2 + 2\alpha}}{\sqrt{\alpha}}
$$
In the following, we choose $\theta = \frac{-\lambda + \sqrt{\lambda^2 + 2\alpha}}{\sqrt{\alpha}} $, whence $\theta > 0$ .
\begin{eqnarray*}
E^Q_{r_0}[e^{-\tilde{a} T_v - \int_0^{T_v} r_s ds}] = e^{\frac{\theta}{\sqrt{\alpha}}(v -{r}_0)}\,E^{Q*}_{(r_0 + \frac{\beta}{\alpha})}
\left[
e^{-\left(\tilde{a}-\frac{\beta}{\alpha} + \frac{\theta}{\sqrt{\alpha}}(\phi + \frac{\lambda \beta}{\alpha})\right)T_v}
\right]\,.
\end{eqnarray*} 
Let $\tilde{\gamma} = \tilde{a}-\frac{\beta}{\alpha} + \frac{\theta}{\sqrt{\alpha}}(\phi + \frac{\lambda \beta}{\alpha})$, then :
\begin{eqnarray*}
E^Q_{r_0}[e^{-\tilde{a} T_v - \int_0^{T_v} r_s ds}] = e^{\frac{\theta}{\sqrt{\alpha}}(v - {r}_0)}\,E^{Q*}_{(r_0 + \frac{\beta}{\alpha})}(e^{-\tilde{\gamma}T_v}).
\end{eqnarray*}
Under the probability $Q^*$, the instantaneous rate $\tilde{r}(t)$ satisfies
\begin{equation}\label{eq.3.2}
d \tilde{r}(t) = (\phi + \frac{\lambda \beta}{\alpha}-(\lambda + \theta\sqrt{\alpha})\tilde{r}(t))\,dt + \sqrt{\alpha \tilde{r}(t)} \, d w^*(t) ,
\end{equation}
where $w^*(t)$ is a $(Q^*, \F_t)$ Brownian motion, with
$$ d w^*(t) = d w(t) + \theta\sqrt{\tilde{r}(t)}\,dt\,.
$$
By applying Theorem \ref{L.S}, we get that
$$
\tilde{\A}_{*} = (\phi + \frac{\lambda \beta}{\alpha}-(\lambda + \theta\sqrt{\alpha})\tilde{r})\frac{d}{d\tilde{r}} + \frac{\alpha \tilde{r}}{2}\frac{d^2}{d\tilde{r}^2}\,,
$$
is the infinitesimal generator associed on \eqref{eq.3.2}. Let us denote
$$
\tilde{\lambda} = \lambda + \theta\sqrt{\alpha}
$$
and
$$
\tilde{\phi}_{0} = \phi - \frac{\beta\theta}{\sqrt{\alpha}}\,.
$$
Then
$$
\tilde{\A}_{*} = (\tilde{\phi}_{0} + \frac{\tilde{\lambda}\beta}{\alpha} - \tilde{\lambda}\tilde{r})\,\frac{d}{d\tilde{r}} + \frac{\alpha \tilde{r}}{2}\frac{d^2}{d\tilde{r}^2}\,.
$$
Here, we take $\lambda > 0$ and $\phi + \frac{\lambda \beta}{\alpha} > 0$, therefore, by Theorem \ref{L.S} applied to $\tilde{r}_t$ under $Q^*$,
$$
E^{Q*}_{(r_0 + \frac{\beta}{\alpha})}(e^{-\tilde{\gamma} T_v}) = \frac{M(\frac{\tilde{\gamma}}{\tilde{\lambda}},(\tilde{\phi}_{0} + \frac{\tilde{\lambda}\beta}{\alpha})\frac{2}{\alpha} , \frac{2\tilde{\lambda}(r_0+\beta/\alpha)}{\alpha})}{M(\frac{\tilde{\gamma}}{\tilde{\lambda}},(\tilde{\phi}_{0} + \frac{\tilde{\lambda}\beta}{\alpha})\frac{2}{\alpha} , \frac{2\tilde{\lambda}(v + \beta/\alpha)}{\alpha})}\,.
$$
Whence,
$$
E^Q_{r_0}[e^{-\tilde{a} T_v - \int_0^{T_v} r_s ds}] = e^{\frac{\theta}{\sqrt{\alpha}}(v - {r}_0)}\, \frac{M(\frac{\tilde{\gamma}}{\tilde{\lambda}},(\tilde{\phi}_{0} + \frac{\tilde{\lambda}\beta}{\alpha})\frac{2}{\alpha} , \frac{2\tilde{\lambda}(r_0+\beta/\alpha)}{\alpha})}{M(\frac{\tilde{\gamma}}{\tilde{\lambda}},(\tilde{\phi}_{0} + \frac{\tilde{\lambda}\beta}{\alpha})\frac{2}{\alpha} , \frac{2\tilde{\lambda} (v + \beta/\alpha)}{\alpha})}\,.
$$
Also,
$$
E^Q_v[\int_0^\infty dTe^{-\tilde{a} T-\int_0^T r_s ds}] = \int_0^\infty dT E^Q_v[e^{-\tilde{a} T-\int_0^T r_s ds}]
= \int_0^\infty dTe^{-\tilde{a}T} B_v(0,T),
$$
where $B_v(0,T)$ denotes the price of a discount bond corresponding to the interest rate process starting from $v$.
\end{proof}

The following result allows us to study the variation of the option price according to the strike $k$.
\begin{corollary}\label{corollary3.3}
We have
\begin{eqnarray*}
\frac{\partial}{\partial k} C(\sup_{u\in[0,T]} r_u , 0 , T , k) = -e^{\frac{\theta }{\sqrt{\alpha}}(k - r_0)}\,\frac{M(\frac{\tilde{\gamma}}{\tilde{\lambda}},(\tilde{\phi}_{0} + \frac{\tilde{\lambda}\beta}{\alpha})\frac{2}{\alpha} , \frac{2\tilde{\lambda}(r_0+\beta/\alpha)}{\alpha})}{{M(\frac{\tilde{\gamma}}{\tilde{\lambda}},(\tilde{\phi}_{0} + \frac{\tilde{\lambda}\beta}{\alpha})\frac{2}{\alpha} , \frac{2\tilde{\lambda}(k + \beta/\alpha)}{\alpha})}}\,B_k(0,T).
\end{eqnarray*}
\end{corollary}
\begin{proof}
Indeed,
$
\frac{\partial}{\partial k} U_{(k,r_0)}(\tilde{a}) = \int_0^\infty e^{-\tilde{a} T} \frac{\partial}{\partial k} C(\sup_{u\in[0,T]} r_u , 0 , T , k)\,dT
$
and
\begin{eqnarray*}
\frac{\partial}{\partial k} U_{(k,r_0)}(\tilde{a}) &=& -e^{-\frac{\theta r_0}{\sqrt{\alpha}}}\,\frac{M(\frac{\tilde{\gamma}}{\tilde{\lambda}},(\tilde{\phi}_{0} + \frac{\tilde{\lambda}\beta}{\alpha})\frac{2}{\alpha} , \frac{2\tilde{\lambda}(r_0+\beta/\alpha)}{\alpha})}{{M(\frac{\tilde{\gamma}}{\tilde{\lambda}},(\tilde{\phi}_{0} + \frac{\tilde{\lambda}\beta}{\alpha})\frac{2}{\alpha} , \frac{2\tilde{\lambda}(k + \beta/\alpha)}{\alpha})}}\, e^{\frac{\theta k}{\sqrt{\alpha}}} P_{\tilde{a}}(k)\\
&=&-\int_0^\infty e^{-\tilde{a}T} e^{\frac{\theta }{\sqrt{\alpha}}(k - r_0)}\,\frac{M(\frac{\tilde{\gamma}}{\tilde{\lambda}},(\tilde{\phi}_{0} + \frac{\tilde{\lambda}\beta}{\alpha})\frac{2}{\alpha} , \frac{2\tilde{\lambda}(r_0+\beta/\alpha)}{\alpha})}{{M(\frac{\tilde{\gamma}}{\tilde{\lambda}},(\tilde{\phi}_{0} + \frac{\tilde{\lambda}\beta}{\alpha})\frac{2}{\alpha} , \frac{2\tilde{\lambda}(k + \beta/\alpha)}{\alpha})}}\,B_k(0,T)\,dT\,.
\end{eqnarray*}
Whence Corollary \ref{corollary3.3} by the injectivity of the Laplace transform.
\end{proof}
\begin{corollary}\label{corollary3.4}
We have also
\begin{eqnarray*}
\frac{\partial}{\partial K} C(\sup_{u\in[0,T]} Y(u, u + \tau) , 0 , T , K) = \frac{\partial}{\partial k} C(\sup_{u\in[0,T]} r_u , 0 , T , k)\,.
\end{eqnarray*}
\end{corollary}
\begin{proof}
Indeed,
\begin{eqnarray*}
\frac{\partial}{\partial K} C(\sup_{u\in[0,T]} Y(u, u + \tau) , 0 , T , K)
&=&\frac{A(\tau)}{\tau}\,\frac{\partial}{\partial K}C(\sup_{u\in[0,T]} r_u,0,T,k)\\
&=&\frac{A(\tau)}{\tau}\,\frac{\partial}{\partial k} C(\sup_{u\in[0,T]} r_u , 0 , T , k)\,\frac{\partial k}{\partial K}\\
&=& \frac{A(\tau)}{\tau}\,\frac{\partial}{\partial k} C(\sup_{u\in[0,T]} r_u , 0 , T , k)\,\frac{\tau}{A(\tau)}\\
&=&\frac{\partial}{\partial k} C(\sup_{u\in[0,T]} r_u , 0 , T , k)\,.
\end{eqnarray*}
Whence the corollary \ref{corollary3.4}.
\end{proof}

Finally, we have an explicit formula of the option price derivative on the maximum of the yield with respect to the strike $K$: 
\begin{corollary}\label{corollary3.5}
\begin{multline*}
\frac{\partial}{\partial K} C(u\in\sup_{[0,T]} Y(u, u + \tau) , 0 , T , K) =\\ -e^{\frac{\theta}{\sqrt{\alpha}}( \frac{\tau K -b(\tau)}{A(\tau)} - r_0)}\,\frac{M(\frac{\tilde{\gamma}}{\tilde{\lambda}},(\tilde{\phi}_{0} + \frac{\tilde{\lambda}\beta}{\alpha})\frac{2}{\alpha} , \frac{2\tilde{\lambda}(r_0+\beta/\alpha)}{\alpha})}{{M(\frac{\tilde{\gamma}}{\tilde{\lambda}},(\tilde{\phi}_{0} + \frac{\tilde{\lambda}\beta}{\alpha})\frac{2}{\alpha} , \frac{2\tilde{\lambda}(\frac{\tau K -b(\tau)}{A(\tau)} + \beta/\alpha)}{\alpha})}}\,B_{(\frac{\tau K -b(\tau)}{A(\tau)})}(0,T)\,.
\end{multline*}
\end{corollary}
\section{appendix: Novikov's condition}
In this section, we check Novikov's condition which allows us to apply Girsanov's Theorem.
\begin{lemma}\label{Le}
$E^Q_{\tilde{r}_0}(e^{\frac{1}{2}\int_0^T\theta^2 \tilde{r}_s ds}) < +\infty$ .
\begin{proof}
Indeed, for all $a > 0$, $E^Q_{\tilde{r}_0}(e^{\int_0^T a \tilde{r}_s ds}) = E^Q_{\tilde{r}_0}(e^{-\int_0^T(-a)\tilde{r}_s ds})$.

Let ${r'}_s = -a\,\tilde{r}_s < 0$ , then ${r'}(t)$ satisfies
\begin{equation}
d{r'}(t) = (\phi'-\lambda {r'}(t))dt + \sqrt{\alpha' {r'}(t)} d{w'}(t)\,,
\end{equation}
with $\phi' = -a(\phi + \frac{\lambda \beta}{\alpha})$, $\alpha' = -\alpha$ et ${w'}(t) = -w(t)$.

Denote
$$
\mu' = \lambda\sqrt{1 + \frac{2\alpha'}{\lambda^2}}\,
$$
$$
k' = \frac{\mu' - \lambda}{\mu' + \lambda}\,,
$$
then
\begin{eqnarray*}
E^Q_{\tilde{r}_0}(e^{\int_0^T a \tilde{r}_s ds}) &=& E^Q_{{r'}_0}(e^{\int_0^T (-{r'}_s) ds})\\ 
&=& B_{{r'}_0}(T)\\
&=& e^{A'(T){r'}_0 + b'(T)}.
\end{eqnarray*}
The issue here is to show that $\lim\limits_{T \rightarrow +\infty}\ (A'(T){r'}_0 + b'(T)) < +\infty$.
$\bigskip$

If $\lambda > 0$, then $A'(T) = \displaystyle\frac{1 + k'}{\mu'}\frac{1 - e^{-\mu'T}}{1 + k'e^{-\mu'T}}$, $b'(T) = \displaystyle\frac{2\phi'}{\alpha'}\log(\frac{1 + k'e^{-\mu'T}}{1 + k'}) + \displaystyle\frac{\phi'(\mu' - \lambda)}{\alpha'}T$, and
$$
\lim_{T \rightarrow +\infty}\ \left(A'(T){r'}_0 + b'(T)\right) < +\infty \text{ iff } \phi + \frac{\lambda \beta}{\alpha} > 0.
$$
If $\lambda < 0$, then $b'(T) = \displaystyle\frac{2\phi'}{\alpha'}\log\left\lvert\frac{1 + k'e^{-\mu'T}}{1 + k'}\right\rvert + \frac{\phi'(\mu' - \lambda)}{\alpha'}T$, and
$$
\lim_{T \rightarrow +\infty}\ (A'(T){r'}_0 + b'(T)) < +\infty \text{ iff } \phi + \frac{\lambda \beta}{\alpha} < 0.
$$
In these cases, it suffices to take $a = \frac{\theta^2}{2}$ in order to obtain Novikov's condition.
\end{proof}
\end{lemma}
\begin{rmq}
$B_{{r'}_0}(T)$ {\em is the price of discount bond at time} $T$ {\em worth one euro at time} $t=0$, {\em for the instantaneous rate given by }${-r'}_t = a\,\tilde{r}_t.$ 
\end{rmq}
\section{Appendix : Simulation}
We illustrate in this section the valuation of the European call option on the maximum of the instantaneous rate in the one factor affine mode bu using its explicite formula of the Laplace's transform given in the proposition \ref{prop3.2}. The functions and the procedure which are used are programmed in Maple and their codes are given in the following section. 
\subsection{Numerical Simulations}
\begin{center}

$\bold{Table1}$. Parameters for European call option in the one factor affine model

\begin{tabular}{c c c c c c c c}
\hline
Case & $\phi$ & $\lambda$ & $\alpha$ & $\beta$ & $T$ & $\tau$ & $K$\\

\hline
$1$ & $0.02$ & $0.2$ & $0.02$ & $0.002$ & $1$ & $10$ & $0.1$\\

$2$ & $0.02$ & $0.2$ & $0.02$ & $0.002$ & $1$ & $0.25$ & $0.1$\\

$3$ & $0.02$ & $0.2$ & $0.02$ & $0.002$ & $0.25$ & $10$ & $0.1$\\

$4$ & $0.02$ & $0.2$ & $0.02$ & $0.002$ & $0.25$ & $0.25$ & $0.1$\\

$5$ & $0.02$ & $0.2$ & $0.0002$ & $0.00002$ & $1$ & $10$ & $0.1$\\

$6$ & $0.02$ & $0.2$ & $0.0002$ & $0.00002$ & $1$ & $0.25$ & $0.1$\\

$7$ & $0.02$ & $0.2$ & $0.0002$ & $0.00002$ & $0.25$ & $10$ & $0.1$\\

$8$ & $0.02$ & $0.2$ & $0.0002$ & $0.00002$ & $0.25$ & $0.25$ & $0.1$\\
\hline
\end{tabular}
\end{center}

\begin{center}

$\bold{Table2}$. Prices of European call option on maximum in the one factor affine model

\begin{tabular}{c c c c c c c c}
\hline
Case & $n=m=5$ & $n=m=6$ & $n=m=7$ & $n=m=8$ & $n=m=9$ & $n=m=10$\\

\hline
$1$ & $0.0013974$ & $0.0018056$ & $0.0026666$ & $0.0035363$ & $0.0037346$ & $0.0045306$\\

$2$ & $0.0035322$ & $0.0045639$ & $0.0067402$ & $0.0089386$ & $0.0094398$ & $0.0114518$\\

$3$ & $0.0003312$ & $0.0004357$ & $0.0005177$ & $0.0005772$ & $0.0006170$ & $0.0010445$\\

$4$ & $0.00083718$ & $0.0011014$ & $0.0013088$ & $0.0014589$ & $0.0015595$ & $0.0026402$\\

$5$ & $0.0003745$ & $0.0004252$ & $0.0006730$ & $0.0009531$ & $0.0009302$ & $0.0011855$\\

$6$ & $0.0008460$ & $0.0009607$ & $0.0015204$ & $0.0021532$ & $0.0021014$ & $0.0026781$\\

$7$ & $0.0001336$ & $0.0001634$ & $0.0001853$ & $0.0002002$ & $0.0002093$ & $0.0003791$\\

$8$ & $0.0003018$ & $0.0003691$ & $0.0004187$ & $0.0004523$ & $0.0004729$ & $0.0008564$\\
\hline
\end{tabular}
\end{center}

\subsection{Program}

The functions and the procedures implemented in Maple are given below :

\begin{verbatim}

Digits := 10;

# the parameters of our model are given by :
p:=10; r_{0}:=0.1; phi:=0.02; lambda:=0.2; alpha:=0.02; beta:=0.002; k:=0.1; 
h:=evalf(exp(-k))

theta:=evalf((-lambda + sqrt(lambda^2+2*alpha))/sqrt(alpha))
phitild:=phi + lambda*beta/alpha

lamdatild:=evalf(lambda + theta*sqrt(alpha))
gamatild:= atild - beta/alpha + evalf(theta*(phi+lambda*beta/alpha)/sqrt(alpha))

mu:=lambda*sqrt(1+2*alpha/lambda^2)
K:=simplify((mu-lambda)/(mu+lambda))

A(T):={1+K}/{mu}*(1-exp(-mu*tau))/(mu*(1+K*exp(-mu*tau))) 

b(T):= 2*phitild*log((1+K*exp(-mu*tau))/(1+K))/alpha
+(phitild*(mu-lambda)/alpha-beta/alpha)*tau + beta*A(tau)/alpha

Kmmer:=convert(series(KummerM(C,B,z),z,p), polynom)   # here, p is fixed on 10

# We do a change of variable: u= exp(-v) and w=exp(-T) 
in the double integral to discritiser the domain.

f(u,w):= exp(-theta*log(u)/sqrt(alpha)+atild*log(w) + A(-log(w))*log(u)
-b(-log(w)))/expand(subs(C = gamatild/lamdatild, B = 2*phitild/alpha,
z = 2*(-log(u)+beta/alpha)*lamdatild/alpha, Kmmer))*(u*w)

Doubleintegral:=proc (a,b,n,c,d,m::integer,f) 
local i, j, delta1, delta2, xx, yy, F, k; 
delta1:=(b-a)/n; 
delta2:=(d-c)/m; 
xx:=vector[row](n+1); 
yy:=vector[row](m+1); 
F:=array(1..n+1,1..m+1); 
for i to n+1 do
xx[i]:= a+(i-1)*delta1 
od; 
for j to m+1 do 
yy[j]:= c+(j-1)*delta2 
od; 
for i from 2 to n+1 do 
for j from 2 to m+1 do 
F[i,j]:= expand(simplify(expand(f(xx[i],yy[j])))); 
print(i,j);
od; od; 
evalf(F); 
k:= simplify(F[2,2]+F[2,m+1]+F[n+1,2]+F[n+1,m+1]); 
for i from 3 to n do 
k:= simplify(k+2*F[i,2]+2*F[i,m+1])
od; 
for j from 3 to m do 
k:= simplify(k+2*F[2,j]+2*F[n+1,j])
od; 
for i from 3 to n do 
for j from 3 to m do 
k:= simplify(k+4*F[i,j])
od; od; 
k:= (1/4)*k*delta1*delta2; 
simplify(k); 
end;

valdbint:=Doubleintegral(0,h,5,0,1,5,f):  # here, we have: n=m=5

valdbint:=factor(valdbint)

lapcall := expand(subs(C = gamatild/lamdatild, B = 2*phitild/alpha, 
z = 2*(r_{0}+beta/alpha)*lamdatild/alpha, Kmmer))*valdbint

with(inttrans):

invlaplace(lapcall,atild,1) # here, we choose for example, the time of maturity T= 1

C(supr_{u}, 0, k, T=1):= exp(-theta*r_{0}/sqrt(alpha))*invlaplace(lapcall,atild,1)

C(supY(u,u + tau) := (A(10)/10)*C(supr_{u}, 0, k, T=1) # here, we take tau = 10
\end{verbatim}

\end{document}